\documentclass[12pt]{article}

\usepackage{amsmath}
\usepackage{amsthm}
\usepackage{amsfonts}
\usepackage{amssymb}
\usepackage{latexsym}
\usepackage{amscd}
\usepackage{stmaryrd}
\usepackage{amsbsy}

\allowdisplaybreaks
\newcommand{\C}{{\mathbb C}}
\newcommand{\Z}{{\mathbb Z}}

\newcommand{\slt}{\mathfrak{sl}_2}

\newcommand{\g}{\mathfrak{g}}

\numberwithin{equation}{section}

\theoremstyle{plain}
\newtheorem{thm}{Theorem}[section]
\newtheorem{cor}[thm]{Corollary}

\newtheorem{prop}[thm]{Proposition}

\begin{document}

\title{Confluent primary fields in the conformal field theory }
\date{\today}
\author{Hajime Nagoya\footnote{Research Fellow of the Japan Society for the Promotion of Science} \\ 
Department of Mathematics, Kobe University,\\
 Kobe 657-8501, Japan
\\
e-mail: nagoya@math.kobe-u.ac.jp
\\
and 
\\
Juanjuan Sun
\\
Graduate School of Mathematical Sciences, The
University of Tokyo, \\Tokyo 153-8914, Japan
\\
email: sunjuan@ms.u-tokyo.ac.jp
}

\date{}

\maketitle

{\footnotesize
\begin{quote}

{\bf Abstract:} 
 For any complex simple Lie algebra, 
we generalize the  primary fileds  in the Wess-Zumino-Novikov-Witten conformal 
field theory for a case with irregular singularities. We refer to these generalized 
primary fields as
 confluent primary 
fields. 
We present the screening currents Ward identity,  a recursion rule for computing the 
 expectation values of the products
of  confluent primary fields. 
In the case of $\slt$,  the expectation values of the products of 
confluent primary fields are integral formulas of solutions to confluent KZ equations 
given in  \cite{JNS}.  
By computing the operator
product expansion of the energy-momentum tensor $T(z)$ and the
confluent primary fields, we obtain new differential operators. 
 Moreover, in the case of $\slt$, 
these differential operators are  the same as those of 
the confluent KZ equations  \cite{JNS}. 

\end{quote} }

{\footnotesize
\begin{quote}

{\noindent\bf Mathematics Subject Classifications (2000):}
32G34, 17B80, 34M55, 37J35, 81R12.

\end{quote} }

{\footnotesize
\begin{quote}

{\noindent\bf Keywords:} KZ equation, Wakimoto realization, Conformal field theory.

\end{quote} }

\section{Introduction}
In  the last twenty-five years, 
 many mathematicians and physicists have contributed to the study and development of
the two-dimensional conformal field theory  (CFT).  
This theory finds many applications 
in statistical physics and  string theory. It also finds applications in 
many fields of mathematics, 
such as representation theory, integrable systems, and 
topology. 

 In the Wess-Zumino-Novikov-Witten (WZNW) 
model,  the Knizhnik-Zamolodchikov (KZ) equation plays a key role. The KZ equation is satisfied by the correlation 
functions of the model \cite{KZ} and 
 a system of linear partial differential equations with regular
singularities. Accordingly, solutions of the KZ equation are expressed in terms of integral representations 
of hypergeometric functions  in several variables  
\cite{SV}. Further, the KZ equation can be viewed as a quantization of 
the Schlesinger equation; this equation defines the isomonodromy deformation of 
 linear differential equations 
with regular singularities \cite{R}, \cite{H}, \cite{Takasaki}. 

The extensions of  the KZ equations to  irregular
singularities have been considered for the following cases.  
Generalized KZ
equations with Poincar\'e rank 1 at infinity were first obtained for $\slt$ in \cite{BK}   and later for any simple Lie
algebra in \cite{FMTV}. For an arbitrary Poincar\'e rank,   confluent KZ
equations for $\slt$ were presented in \cite{JNS}.  The authors  have presented 
confluent KZ equations for $\frak{sl}_N$ with Poincar\'e rank 2 at
infinity in \cite{NS}. 
 In all the above mentioned cases, 
 the solutions
to the equations are represented in terms of  integral
representations of hypergeometric functions of the confluent type.

 Integral representations of 
solutions to the KZ equations have  also been constructed 
using  free field realizations and 
Wakimoto modules \cite{ATY}, \cite{FF},
\cite{IK}, \cite{K}, and \cite{W}. 
Free field realizations provide a clear understanding of the WZNW CFT and also 
provide valuable insights on representation theory and 
quantum field theory. The solutions to confluent KZ equations, however, 
have not been constructed using free field realizations. 

%

In this paper, 
 using free field realizations, 
we generalize primary fields in  the WZNW CFT for a case with irregular singularities for any simple Lie algebra $\g$. Hereafter, 
we refer to these generalized primary fields  as 
confluent primary fields. 
We present the screening currents Ward identity,  a recursion rule for computing the 
 expectation values of the products
of  confluent primary fields. 
As  seen in section 5, the integral representations for the 
expectation values of the products of 
confluent primary fields are multi-variable confluent hypergeometric functions. 
 In the case of $\slt$, these integral representations 
are equivalent to the solutions of confluent KZ equations \cite{JNS} 
obtained through the confluent process from solutions of the KZ equations. 

Furthermore, we compute the operator product expansion (OPE) of
 the energy-momentum tensor $T(z)$ and the confluent primary fields.
Consequently, new differential operators corresponding to the Virasoro
operators $L_{-1},L_{0},\ldots, L_{r-1}$ appear. Note that in the case of the KZ equations,
essentially  all Virasoro operators correspond to the differential
operator $\partial / \partial z$. Moreover, in the case of $\slt$, these differential operators are the same as 
those of the confluent KZ equations \cite{JNS}. 

In addition to free field realizations and OPE,  we use  
 the truncated Lie algebra $\g_{(r)}=\g[t]/t^{r+1}\g[t]$, where the nonnegative 
integer $r$ corresponds to the Poincar\'e rank, and a confluent Verma module 
of $\g_{(r)}$. 
Confluent Verma modules have been 
defined in \cite{JNS}, and these are  natural generalizations of standard Verma modules. 
In addition, the confluent Verma modules correspond to the non-highest weight 
representations of affine Lie algebras in \cite{FFT}. 


The remainder of this paper is organized as follows. In section 2, we introduce
confluent Verma modules and recapitulate Wakimoto realization of the
affine Lie algebra, following \cite{A}, \cite{ATY}, and \cite{Yamada}. 
 In section 3, we define  confluent primary fields
and  we
show that the set consisting of these fields is a $\g_{(r)}$-module. 
In  section 4, we compute the OPE of the energy-momentum tensor $T(z)$ and 
a confluent primary field. 
In the final section, we present integral representations  
of hypergeometric functions of the confluent type 
from the WZNW CFT. Moreover, we write down those integral representations  
for the case of $\slt$,  and we see that they coincide with solutions of the 
confluent KZ equations for $\slt$ \cite{JNS}. 


\section{Preliminary}
Let $\g$ be a complex simple Lie algebra with Cartan subalgebra $\frak{h}$
 and let  $\Delta$ and 
$\Delta_+$ be the set of roots and positive roots,
respectively. We denote  the
Chevalley generators by $e_i$, $f_i$, and $h_i$ 
 and the simple roots  by $\alpha_i$ ($i=1,\ldots,l=\rm{rank} \g$). 
Let $A=(a_{ij})$ be the Cartan matrix of $\g$.  The Cartan matrix is
realized as $a_{ij} = (\nu_i ,\alpha_j )$, where
$\nu_i=\frac{2}{\alpha_i^2}\alpha_i$
 is the coroot and $(,)$, the
symmetric bilinear form.


\subsection{}
 The confluent (highest weight) Verma modules in the case of 
 $\g=\frak{sl}_2$ have been introduced in \cite{JNS}. In this section, we generalize them to 
 the case of a complex simple Lie algebra. 
 
 To describe an irregular singularity, 
we use the truncated Lie algebra $\g_{(r)}=\g[t]/t^{r+1}\g[t]$ for
$r\in\Z_{\ge 0}$.   We denote
$x\otimes t^i$ by $x[i]$. Let $\frak{b}=
\bigoplus_{\alpha\in\Delta_{+}}  \C e_{\alpha} \oplus \frak{h}$ and $\frak{b}_{(r )} = \frak{b}[t ]/ t^{r+1}
\frak{b}[t ]$. For an ($r + 1$)-tuple of weights $ \lambda = (
\lambda_{0}, \cdots,  \lambda_{r-1},  \lambda_{r} ) $ with a regular element 
$\lambda_{r}$, 
 we define a one-dimensional 
$\frak{b}_{(r)}$-module $\C v_{ \lambda}$ by
\begin{equation}
e_{\alpha}[p]v_{ \lambda}=0,\quad h[p]v_{ \lambda}=\lambda_{p}(h)v_{ \lambda}
\quad (0\le p\le r, \ h\in\frak{h}).
\end{equation}
We denote $ \lambda_{p}(h_{i})$ by $ \lambda^{i}_{p}$. The parameters $ \lambda^{i}_{1},\ldots,  \lambda^{i}_{r}$ are the
new variables corresponding to the irregular singularities.

Consider the induced module 
 \begin{equation}
M( \lambda )
= \mathrm{Ind}^{\g_{(r)}}_{\frak{b}_{(r)}} \C v_{ \lambda}, 
\end{equation}
hereafter called the confluent Verma module.  



\subsection{}
Following \cite{A}, \cite{ATY}, \cite{IK}, and \cite{Yamada}, we recall free field
realizations for  simple Lie algebras. Let $\beta_\alpha(z)$ and
$\gamma^\alpha(z)$ ($\alpha\in \Delta_+$) be boson operators with
conformal weights 1 and 0 that satisfy the canonical OPE
\begin{equation}
\beta_\alpha(z)\gamma^\beta(w)=\frac{\delta_{\alpha,\beta}}{z-w}+\cdots, 
\end{equation}
where the dots denote the terms that are regular at $z=w$. 
We also introduce a free boson $\varphi(z)$ taking value in the Cartan subalgebra and
$\varphi_i(z)=(\nu_i,\varphi(z))$ ($i=1,\ldots,l$) with the OPE
\begin{equation}\label{eq free boson}
\varphi_i(z)\varphi_j(w)=\frac{(\nu_i,\nu_j)}{\kappa}\log (z-w)+\cdots.
\end{equation}
Note that we have 
\begin{equation}\label{eq free boson d}
\left(\lambda, \frac{\partial^m \varphi(z)}{m!}\right)
\left(\mu, \frac{\partial^n \varphi(w)}{n!}\right)
=\frac{(\lambda, \mu)}{\kappa}
\begin{pmatrix}
m+n
\\
m
\end{pmatrix}
\frac{(-1)^{m+1}}{m+n}\frac{1}{(z-w)^{n+m}}+\cdots, 
\end{equation}
for $m,n\in\mathbb{Z}_{\ge 0}$, $m+n\neq 0$.

 We recall the definition of the currents $E_i(z)$, $H_i(z)$, and $F_i(z)$. 

Let $V(\lambda)$ be the Verma module of $\g$ with the highest
weight vector $| \lambda \rangle$ 
and $V( \lambda)^*$, the dual module of $V( \lambda)$
generated by $\langle \lambda |$ with $\langle  \lambda |f_\alpha=0$
and $\langle \lambda | h_i= \lambda^i$. The bilinear form $\langle,
\rangle$ is defined from $\langle  \lambda | \lambda\rangle=1$. 

We realize the elements of the algebra $\g$ in terms of those in the polynomial ring $\C[x^\alpha]$ with positive roots 
$\alpha\in\Delta_+$ as differential operators. 
The
differential operator $J\left(
\partial/\partial x, x,  \lambda\right)$ corresponding to an element 
$J$ in $\g$ is defined by the following right action 
\begin{equation}
J\left( \frac{\partial}{\partial x}, x,  \lambda\right)
\langle  \lambda| Z=\langle  \lambda |Z J,
\end{equation}
where $Z=\exp\left(\sum_{\alpha\in\Delta_+}x^\alpha
e_\alpha\right)$. 

 The differential operators $E_i$, $H_i$, and $F_i$ corresponding  to 
 the generators $e_i$, $h_i$, and $f_i$ ($i=1,\ldots, l=\mathrm{rank}\ \g$) have the following form:
\begin{align}
&E_i=
\sum_{\alpha\in\Delta_+}E_i^\alpha(x)\frac{\partial}{\partial x^\alpha},
\\
&H_i=
\sum_{\alpha\in\Delta_+}H_i^\alpha(x)\frac{\partial}{\partial x^\alpha}
+ \lambda^i,
\\
 &F_i=\sum_{\alpha\in\Delta_+}F_i^\alpha(x)\frac{\partial}{\partial x^\alpha}
+ \lambda^ix^{\alpha_i},
\end{align}
where $X_i^\alpha(x)$ ($X=E,H,F$) are polynomials in $\C[x^\alpha]$. 
These operators give a highest weight representation of $\g$ on $\C[x^\alpha]$. 
We call this  the differential realization of $\g$. 

In the differential realization, the highest weight vector 
$|\lambda\rangle$ is $1\in\C[x^\alpha]$. For an ordered set
$I=\{\alpha_{i_1},\ldots,\alpha_{i_n}\}$ of simple roots $\alpha_{i_1},\ldots,\alpha_{i_n}$, 
the vectors $P_
\lambda^I \cdot 1=\prod_{k=1}^n F_{i_k}\cdot 1$ form the
basis of the descendants of $1$.
This basis $P_\lambda^I$ is expressed by the expectation value
\begin{equation}\label{eq P r=0}
P_\lambda^I=\langle\lambda |Z\prod_{k=1}^nf_{i_k}|
\lambda\rangle, 
\end{equation}
because, by definition, 
\begin{equation}
F_{i_1}\cdots F_{i_n}\langle \lambda |Z=\langle \lambda |Z f_{i_1}\cdots f_{i_n}. 
\end{equation}

For $i=1,\ldots,l$, let the currents be defined by 
\begin{align}
&E_i(z)=\sum_{\alpha\in \Delta_+}:E_i^\alpha(\gamma(z))\beta_\alpha(z):,
\\
&H_i(z)=\sum_{\alpha\in\Delta_+}  :H_i^\alpha(\gamma(z))  \beta_\alpha(z)   : +a_i(z),
\\
&F_i(z)=\sum_{\alpha\in\Delta_+}: F_i^\alpha(\gamma(z))\beta_\alpha(z)  :+\gamma^{\alpha_i}(z)
a_i(z)+
r_i\partial \gamma^{\alpha_i}(z),\label{eq F(z)}
\end{align}
where $:\cdot:$ stands for the normal ordering and for $X=E,H,F$, 
$X_i^\alpha(\gamma(z))$ ($i=1,\ldots, l$)  
are obtained by replacing $x^\beta$ with $\gamma^\beta(z)$ in $X_i^\alpha(x)$, while   
 $a_i(z)=\kappa \partial \varphi_i(z)$.  
 The currents $E_i(z)$, $H_i(z)$, and $F_i(z)$
satisfy the following OPEs:
\begin{align}
H_i(z)H_j(w)&=\frac{k(\nu_i,\nu_j)}{(z-w)^2}+\cdots,
\\
H_i(z)E_j(w)&=\frac{a_{ij}}{z-w}E_j(w)+\cdots,
\\
H_i(z)F_j(w)&=-\frac{a_{ij}}{z-w}F_j(w)+\cdots,
\\
E_i(z)F_j(w)&=\frac{k\delta_{i,j}}{(z-w)^2}+\frac{\delta_{i,j}}{z-w}H_i(w)+\cdots; 
\end{align}
these give the level $k$ Wakimoto realization. The coefficients 
$r_i$ of $\partial \gamma^{\alpha_i}$ in \eqref{eq F(z)} are also determined (see
\cite{A} and \cite{ATY}, for example).

\subsection{}

We also introduce differential operators for the screening currents.
Let $S_\alpha$ be defined by
\begin{equation}
S_\alpha\left(\frac{\partial}{\partial x},x\right)\langle \lambda |Z=\langle \lambda| e_\alpha Z
\end{equation}
with
\begin{equation}
S_\alpha\left(\frac{\partial}{\partial x},x\right)
=\sum_{\beta\in\Delta_+}S_\alpha^\beta(x)\frac{\partial}{\partial x^\beta}, 
\end{equation}
for some polynomials $S_\alpha^\beta(x)\in\C[x]$. 

The energy-momentum tensor $T(z)=\sum_{n\in\Z}L_nz^{-n-2}$ is realized as
\begin{equation}
T(z)=\sum_{\alpha\in\Delta_+}
:\beta_{\alpha}(z)\partial\gamma^\alpha(z):+\sum_{i=1}^l :\frac{\kappa}{2}
\partial\varphi_i(z)\partial\varphi^i(z)-
\rho_i\partial^{2}\varphi^i(z) :,
\end{equation}
where $\rho=\frac{1}{2}\sum_{\alpha\in\Delta_+}\alpha$. This is
also realized by the Sugawara construction \cite{Sugawara}
\begin{equation}
T(z)=\frac{1}{2\kappa}: \sum_{i=1}^lH_i(z)H^i(z)+\sum_{\alpha\in\Delta_+}
\frac{\alpha^2}{2}\left(E_\alpha(z)F_\alpha(z)+F_\alpha(z)E_\alpha(z)\right) :.
\end{equation}
The operators $L_n$ ($n\in\mathbb{Z}$) generate the Virasoro algebra.

Let us introduce the screening currents
\begin{equation}\label{eq screening currents}
s_i(z)=S_i(z):
 e^{-\alpha_i\varphi(z)} :,
\end{equation}
where $S_i(z)=:\sum_{\beta\in\Delta_+}
S^\beta_{\alpha_i}\left(\gamma(z)\right)\beta_{\beta}(z):$ 
and the screening operators are given as 
\begin{equation}
Q_i=\oint  s_i(z) dz .
\end{equation}
Then,  we have the following proposition (see \cite{A}, \cite{ATY}, and \cite{Yamada}, 
for example).
\begin{prop}\label{screening prop}
The products $E_i(z)s_j(w)$ and $H_i(z)s_j(w)$ are regular at $z=w$, and the products of $F_i(z)s_j(w)$ and 
$T(z)s_j(w)$ are given as follows. 
\begin{align}
&F_i(z)s_j(w)=\kappa\delta_{i,j}\frac{2}{\alpha_i^2} \frac{\partial}{\partial w}\left( \frac{1}{z-w}:e^{-
\alpha_i\varphi(w)}:   \right)
+\cdots,
\\
&T(z)s_j(w)=\frac{\partial}{\partial w}\left( \frac{1}{z-w}s_j(w)\right) +\cdots.
\end{align}
\end{prop}
The above proposition implies that the screening operators $Q_i$
commute with the currents and the energy-momentum tensor $T(z)$.

\subsection{}
 We consider the differential realization corresponding to the
confluent Verma module. Let $M( \lambda)$ be a confluent
Verma module of $\g_{(r)}$ with weights $(\lambda_0,\lambda_1,\ldots,\lambda_r)$ 
and $M( \lambda)^*$ the
dual of $M( \lambda)$. We replace $Z$ in the regular case with 
$Z=\exp\left(\sum_{i=0}^r\sum_{\alpha\in\Delta_+}
x_i^\alpha e_\alpha[i]\right)$. In the same manner as the regular case, we
define the differential realization of $\g_{(r)}$ on
$\C[x^\alpha_i]$ ($\alpha\in\Delta_+$, $0\le i\le r$)  as    
\begin{equation}
J\left( \frac{\partial}{\partial x}, x,\lambda \right)\langle \lambda | 
Z=\langle \lambda |Z J
\end{equation}
for $J\in\g_{(r)}$. We also introduce an irregular version of $S_\alpha$ as  
\begin{equation}\label{eq screening current p}
S_\alpha[p]\langle \lambda| Z=\langle \lambda | e_\alpha[p]Z \quad (p=0,\ldots,r). 
\end{equation}

The
differential operators $X[p]$ ($X=E_i,H_i,F_i$) 
corresponding to $x[p]\in \g_{(r)}$ ($x=e_i,h_i,f_i$) and $S_\alpha[p]$ are given by 
\begin{align}
&E_i[p]=
\sum_{q=p}^r\sum_{\alpha\in\Delta_+}E_{i,q}^\alpha[p](x)\frac{\partial}{\partial x^\alpha_q},
\\
&H_i[p]=\sum_{q=p}^r
\sum_{\alpha\in\Delta_+}H_{i,q}^\alpha[p](x)\frac{\partial}{\partial x^\alpha_q}
+ \lambda^i_p,
\\
 &F_i[p]=\sum_{q=p}^r
 \sum_{\alpha\in\Delta_+}F_{i,q}^\alpha[p](x)\frac{\partial}{\partial x^\alpha_q}
+ \sum_{q=p}^r\lambda^i_{q}x^{\alpha_i}_{q-p},\label{eq F_i[p]}
\\
&S_\alpha[p]=\sum_{q=p}^r\sum_{\alpha\in\Delta_+}S_{i,q}^\alpha[p](x)\frac{\partial}{\partial x^\alpha_q},
\end{align}
for some polynomials $X^\alpha_{i,q}[p](x)$ that are obtained by replacing
monomials
\begin{equation}
x^{\beta_{1}}\cdots x^{\beta_{m}}\quad (\beta_1,\ldots,\beta_n\in\Delta_+)
\end{equation}
in $X_i^\alpha(x)$  with
\begin{equation}
\sum_{j_1+\cdots+j_n+p=q}
x_{j_1}^{\beta_{1}}\cdots x_{j_m}^{\beta_{m}}.
\end{equation}

For an ordered set
$I=\{(\alpha_{i_1},k_{i_1}),\ldots,(\alpha_{i_n},k_{i_n})\}$ 
($\alpha_{i_1},\ldots,\alpha_{i_n}\in \Delta_+$, $0\le k_{i_1},\ldots,k_{i_n}\le r$), 
we define a polynomial $P_\lambda^I(x)$ 
in $\C[x_i^\alpha]$ as 
\begin{equation}
P^I_ \lambda(x)=\langle  \lambda | Z\prod_{j=1}^n f_{\alpha_{i_j}}[k_j]| \lambda\rangle. 
\end{equation}
Note that we have
\begin{equation}
P^I_\lambda(x)=F_{i_1}[k_{i_1}]\cdots F_{i_n}[k_{i_n}]\cdot 1, 
\end{equation}
because,  by definition, 
\begin{equation}
F_{i_1}[k_{i_1}]\cdots F_{i_n}[k_{i_n}]\langle \lambda |Z=
\langle \lambda |Z f_{i_1}[k_{i_1}]\cdots f_{i_n}[k_{i_n}]
\end{equation}
holds.

\section{Confluent primary field}
In this section, we introduce confluent primary fields; these are   
natural generalizations of primary fields with irregular singularities in the WZNW CFT.  

\subsection{}

For an $(r+1)$-tuple of weights $\lambda=(\lambda_0,\lambda_1,\ldots,\lambda_r)$ with 
 a regular element $\lambda_r$, we set 
 \begin{equation}
v_\lambda(z)=:  \exp \left(
\sum_{i=0}^{r}  \lambda_{i}\frac{\partial^{i}\varphi(z)}{i!}
 \right)         :
\end{equation}
and 
we define a vector space $\mathcal{P}$ over $\C$ generated by 
\begin{equation}
\left\{
P_\lambda^I\left( \gamma(z)\right)
v_\lambda(z)
\right\},
\end{equation}
where $P_\lambda^I\left( \gamma(z)\right)$ is a polynomial of bosons
$\partial^i\gamma^\alpha(z)$ ($\alpha\in\Delta_+$, $0\le i\le r$) that 
are obtained by replacing $x^\alpha_i$ in the polynomial $P_\lambda^I(x)$ for 
an ordered set $I=\{ (\alpha_{i_1},k_{i_1}),\ldots,(\alpha_{i_n},k_{i_n})\}$ 
($\alpha_{i_1},\ldots,\alpha_{i_n}\in \Delta_+$, $0\le k_{i_1},\ldots,k_{i_n}\le r$) 
with $\partial^i \gamma^\alpha(z)/i!$. For an element $X$   ($X=E_i,H_i,F_i$ $i=1,\ldots,l$), 
we define the action of  $X[n]$,  where $X(z)=\sum_{n\in\mathbb{Z}}X[n]z^{-n-1}$,  
on the element
$\Phi(w)\in\mathcal{P}$ as 
\begin{equation}\label{eq def action current}
X[n]\Phi(w)=\oint_w \frac{dz}{2\pi i}(z-w)^nX(z)\Phi(w).
\end{equation}
Since a polynomial $P_\lambda^I( \gamma(z))$ consists of 
$\partial^i\gamma^\alpha(z)$ ($\alpha\in\Delta_+$, $0\le i\le r$) and
$v_\lambda(z)
$ 
consists of $\partial^i\varphi(z)$ ($0\le i\le r$), 
   the OPE between $X(z)$ and $\Phi(w)$ and the definition \eqref{eq def action current} induce 
\begin{equation}\label{eq non highest}
X[n]\Phi(w)=0\quad (n>r). 
\end{equation}
We call elements $\Phi(z)\in\mathcal{P}$ confluent primary fields. 

By \eqref{eq def action current}, the loop algebra $\g\otimes \C[t,t^{-1}]$ acts on the space of 
confluent primary fields $\mathcal{P}$, and 
 the Lie subalgebra $\g\otimes t^{r+1}\C[t]$  annihilates the vectors $\Phi(z)$ of $\mathcal{P}$ 
 because of \eqref{eq non highest}. 
Hence, the vector space $\mathcal{P}$  
  generates a 
non-highest weight representation
 in  \cite{Fed} and \cite{FFT}. Moreover, we have the next proposition.
\begin{prop}
The vector space of confluent primary fields $\mathcal{P}$ is a 
 $\g_{(r)}$-module with the highest weight vector $v_\lambda(z)
 $ 
 such that 
 \begin{equation}
 e_i[p]v_\lambda(z)=0, 
 \quad h_i[p]v_\lambda(z)=\lambda_p^iv_\lambda(z)
 \quad (1\le i\le l,\ 1\le p\le r)
 \end{equation}
 and 
 \begin{equation}\label{eq action of ff}
f_{i_1}[k_{i_1}]\cdots f_{i_n}[k_{i_n}]
v_\lambda(z)
 =P_\lambda^I(\gamma(z))
v_\lambda(z),
\end{equation}
where $I=\{ (\alpha_{i_1},k_{i_1}),\ldots,(\alpha_{i_n},k_{i_n})\}$ 
($\alpha_{i_1},\ldots,\alpha_{i_n}\in \Delta_+$, $0\le k_{i_1},\ldots,k_{i_n}\le r$). 
\end{prop}
\begin{proof}
Since the current $E_i(t)$ does not depend on free bosons
$\varphi(t)$,  the OPE of $E_i(t)$ and $:\exp \left(
\sum_{i=0}^{r}  \lambda_{i}\frac{\partial^{i}\varphi(z)}{i!}
 \right):$ does not have singular parts. Hence, $e_i[p]$ 
 annihilates the element 
 $v_\lambda(z)$. 
 
Computing the OPE of $a_i(t)$ and $:\exp\left( \sum_{i=0}^{r}
 \lambda_{i}\frac{\partial^{i}\varphi(z)}{i!}
 \right):$,  we observe that the elements
$h_i[p]$ ($1\le i\le l$, $0\le p\le r$) act as
$\lambda_p^i$ on $v_\lambda(z)$.  

 We now prove \eqref{eq action of ff}.  By definition, 
we have
\begin{equation}\label{eq action of f_j[p] 1}
  f_{j}[p]
 :\exp \left(
\sum_{i=0}^{r}  \lambda_{i}\frac{\partial^{i}\varphi(z)}{i!}
 \right):=
 \oint_z \frac{dt}{2\pi \sqrt{-1}}\left(t-z\right)^{p}F_{\alpha_{j}}(t)
 :\exp \left(
\sum_{i=0}^{r}  \lambda_{i}\frac{\partial^{i}\varphi(z)}{i!}
 \right):.
 \end{equation}
 From \eqref{eq F(z)}, we only need to compute the OPE of $\gamma^{\alpha_j}(t)a_j(t)$. 
 Hence, 
  the right-hand side of \eqref{eq action of f_j[p] 1} equals  
 \begin{equation}\label{eq action of f_j[p] 2}
 \oint_z \frac{dt}{2\pi \sqrt{-1}}\left(t-z\right)^{p}:
 \sum_{i=0}^r\frac{ \lambda_i^j}{(t-z)^{i+1}}\gamma^{\alpha_j}(t)
 \exp \left(
\sum_{i=0}^{r}  \lambda_{i}\frac{\partial^{i}\varphi(z)}{i!}
 \right):. 
 \end{equation}
 By taking the Taylor expansion of $\gamma^{\alpha_j}(t)$ at $z$, \eqref{eq action of f_j[p] 2} 
 equals 
 \begin{equation}\label{eq action of f_j[p] 3}
	  	 \sum_{i=p}^r \lambda_i^j\frac{\partial^{i-p}\gamma^{\alpha_j}(z)}{(i-p)!}
: \exp \left(
\sum_{i=0}^{r}  \lambda_{i}\frac{\partial^{i}\varphi(z)}{i!}
 \right):. 
 \end{equation}
Therefore, recalling \eqref{eq F_i[p]}, we obtain 
 \begin{equation}
 f_{j}[p]
 :\exp \left(
\sum_{i=0}^{r}  \lambda_{i}\frac{\partial^{i}\varphi(z)}{i!}
 \right):=
 P_ \lambda^{(\alpha_j,p)}(\gamma(z)):\exp \left(
\sum_{i=0}^{r}  \lambda_{i}\frac{\partial^{i}\varphi(z)}{i!}
 \right):.
\end{equation}

 For a polynomial
$g(x)\in\C[x^\alpha_i]$, we can verify 
\begin{equation}
f_{j}[p]g(\gamma(z)): \exp \left(
\sum_{i=0}^{r}  \lambda_{i}\frac{\partial^{i}\varphi(z)}{i!}
 \right):=\left(F_{i}[p]g\right)(\gamma(z))
 : \exp \left(
\sum_{i=0}^{r}  \lambda_{i}\frac{\partial^{i}\varphi(z)}{i!}
 \right):
\end{equation}
 in the same manner as above. Therefore, we obtain
 \begin{align*}
 &f_{i_1}[k_1]\cdots f_{i_n}[k_n]
 :\exp \left(
\sum_{i=0}^{r}  \lambda_{i}\frac{\partial^{i}\varphi(z)}{i!}
 \right):\\
  &=\langle  \lambda |Zf_{i_1}[k_1] 
 \cdots f_{i_n}[k_n]| \lambda\rangle(\gamma(z)) 
 :\exp \left(
\sum_{i=0}^{r}  \lambda_{i}\frac{\partial^{i}\varphi(z)}{i!}
 \right):
 \\
 &=P_ \lambda^I(\gamma(z))
 :\exp \left(
\sum_{i=0}^{r}  \lambda_{i}\frac{\partial^{i}\varphi(z)}{i!}
 \right):.
 \end{align*} 
\end{proof}


\section{Operator product expansion}

In this section, we compute the OPE of the energy-momentum tensor and
confluent primary fields. For $k=0,1,\ldots,r-1$, let $\overline{D}_k$ be
an endomorphism  of $\mathcal{P}$ defined by
\begin{align}
&\overline{D}_k \left(\langle \lambda |Z \prod_{j=1}^n f_{\alpha_{i_j}} |\lambda \rangle(\gamma(z))
:  \exp \left(
\sum_{i=0}^{r}  \lambda_{i}\frac{\partial^{i}\varphi(z)}{i!}
 \right)         :
\right)
\\
&=\langle \lambda |d_k\left(Z\right) \prod_{j=1}^n f_{\alpha_{i_j}} |\lambda \rangle(\gamma(z))
D_k\left\{ : \exp \left(
\sum_{i=0}^{r}  \lambda_{i}\frac{\partial^{i}\varphi(z)}{i!}
 \right)         :\right\}, \nonumber 
\end{align}
where $d_k$ is the derivation given by $d_k(x[p])=px[p+k]$ ($x[p]\in\g_{(r)}$) and 
$D_k$, the differential operator given by 
$D_k=\sum_{i=1}^l\sum_{p=1}^{r-k}p\lambda_{p+k}^i\partial/\partial \lambda_p^i$. 
We note that in the case of $\slt$, the derivations $d_k+D_k$ ($k=0,1,\ldots,r-1$) acting
on the confluent Verma module are the same as those of the confluent KZ equation \cite{JNS}. 

\begin{prop}
The operator product expansion of the energy-momentum tensor $T(z)$
and the confluent primary field $\Phi(w)\in \mathcal{P}$ is given as follows:
\begin{align}
T(z)\Phi(w)=&\sum_{k=0}^{r-1}\frac{1}{(z-w)^{k+2}}\overline{D}_k\Phi(w)
+\frac{1}{z-w}\partial_w \Phi(w)
\\
+&\frac{1}{2\kappa}\left(\sum^r\limits_{p=0}\frac{
\lambda_p}{(z-w)^{p+1}}\right)^2\Phi(w)
+\frac{1}{\kappa}\sum^r\limits_{p=0}\frac{(p+1)(\rho,\lambda_p)}{(z-w)^{p+2}}\Phi(w)+\cdots.\nonumber 
\end{align}
\end{prop}
\begin{proof}
Using the OPE between the bosons, we obtain
\begin{align*}
T(z)\Phi(w)&=\sum^r\limits_{p=0}\frac{ \lambda_p}{(z-w)^{p+1}}:\partial_z\varphi(z)\Phi(w):
\\
&+\frac{1}{2\kappa}\left(\sum^r\limits_{p=0}\frac{ \lambda_p}{(z-w)^{p+1}}\right)^2\Phi(w)
+\frac{1}{\kappa}\sum^r\limits_{p=0}\frac{(p+1)(\rho,\lambda_p)}{(z-w)^{p+2}}\Phi(w)\\
&-\sum\limits_{\alpha\in\Delta_+}\sum^r\limits_{p=0}\frac{e_{\alpha}[p]}{(z-w)^{p+1}}:\partial_z\gamma^{\alpha}(z)\Phi(w):
+\cdots.
\end{align*}
Taking the Taylor expansion at $z=w$ in the above, we obtain 
\begin{align*}
T(z)\Phi(w)&=\sum^r\limits_{k=0}\frac{1}{(z-w)^{k+1}}\left(\sum^r\limits_{p=0}\frac{ \lambda_{p+k}}{p!}:\partial^{p+1}_w\varphi(w)\Phi(w):\right.
\\
&\quad\quad\left.-\sum\limits_{\alpha\in\Delta_+}\sum^r\limits_{p=0}\frac{e_{\alpha}[p+1]}{p!}:(\partial^{p+1}_w\gamma^{\alpha}(w))\Phi(w):\right)
\\
&+\frac{1}{2\kappa}\left(\sum^r\limits_{p=0}\frac{
\lambda_p}{(z-w)^{p+1}}\right)^2\Phi(w)
+\frac{1}{\kappa}\sum^r\limits_{p=0}\frac{(p+1)(\rho,\lambda_p)}{(z-w)^{p+2}}\Phi(w)+\cdots.
\end{align*}
On the other hand, we have
\begin{align*}
\partial_w\Phi(w)&=\sum^r\limits_{p=0}\frac{ \lambda_p}{p!}:\partial^{p+1}_w\varphi(w)\Phi(w):
-\sum\limits_{\alpha\in\Delta_+}\sum^r\limits_{p=0}\frac{e_{\alpha}[p]}{p!}:(\partial_w^{p+1}\gamma^{\alpha}(w))\Phi(w): 
\end{align*}
and
\begin{align*}
\overline{D}_k\Phi(w)&=\sum^r\limits_{p=0}\frac{ \lambda_{p+k+1}}{p!}:\partial^{p+1}_w\varphi(w)\Phi(w):\\
&-\sum\limits_{\alpha\in\Delta_+}\sum^r\limits_{p=0}\frac{e_{\alpha}[p+k+1]}{p!}:(\partial_w^{p+1}\gamma^{\alpha}(w))\Phi(w):\;(k=0,\ldots,r-1). 
\end{align*}
This completes  the proof.
 \end{proof}

\begin{cor}
For $ n\in\Z$, we have
\begin{align}
\left[ L_ n, \Phi(w) \right]
 =&w^{ n+1}\partial_w\Phi(w)+\sum^{r-1}\limits_{k=0}\frac{( n+1)!}{( n-k)!}w^{ n-k}\overline{D}_{k}\label{eq virasoro primary}
 \Phi(w)
 \\
 +&\frac{1}{2\kappa}\sum^{2r}\limits_{k= 0}\sum\limits_{p+q=k,\atop
 p,q\geq0} (\lambda_p, \lambda_q)\frac{( n+1)!}{( n-k)!}w^{ n-k}\Phi(w)\nonumber
 \\
 +&\frac{1}{\kappa}\sum^r\limits_{k=0}(k+1)(\rho,\lambda_k)\frac{( n+1)!}{( n-k)!}w^{ n-k}\Phi(w).\nonumber
 \end{align}
 Here, for $k\geq  n+1$, we set $\frac{( n+1)!}{( n-k)!}=0$.
\end{cor}
We note that when $n=-1$, the relation \eqref{eq virasoro primary} reduces to 
\begin{equation}
\left[ L_{-1}, \Phi(w)\right]=\partial_w\Phi(w). 
\end{equation}


\section{Integral representation}

In this section, following \cite{A}, \cite{ATY}, and \cite{Yamada}, we compute an expectation value of the composition of confluent primary fields multiplied by the screening operators. 
 In the case of $\g=\slt$, 
we see that the integral representations derived from confluent primary fields coincide with 
solutions to the confluent KZ equations for $\slt$ \cite{JNS}.

\subsection{}
For $r_1,\ldots,r_n\in\mathbb{Z}_{\ge 0}$, 
let $\lambda^{(a)}$ be ($r_a+1$)-tuple weights ($(\lambda^{(a)})_0,\ldots,(\lambda^{(a)})_{r_a}$) 
with a regular element $(\lambda^{(a)})_{r_a}$. $P_a(v)$ denotes a polynomial 
$\langle \lambda^{(a)}|Z |v_a\rangle$ for $v_a\in M(\lambda^{(a)})$ and 
$P_a(\gamma(z_a))$,  
a polynomial of $\partial^i \gamma^\alpha(z_a)$ ($i=0,\ldots,r_a$, $\alpha\in\Delta_+$) 
obtained by replacing $x^\alpha_i$ with $\partial^i \gamma^\alpha(z_a)/i!$ in $P_a(v_a)$.   

An integral representation of a hypergeometric function of the confluent type is given by an expectation value 
\begin{equation}
\left \langle\int 
\prod_{i=1}^mdt_i :e^{-\alpha_{\bar{i}}\varphi(t_i) }:S_{\bar{i}}(t_i)
\prod_{a=1}^nP_a\left(\gamma(z_a)\right) 
:  \exp \left(
\sum_{i=0}^{r_a}  (\lambda^{(a)})_{i}\frac{\partial^{i}\varphi(z_a)}{i!}
 \right)         :
\right \rangle, 
\end{equation}
where $\bar{i}$ is an element in $\{1,\ldots,l=\rm{rank}\ \g\}$, $\alpha_{\bar{i}}$ ($i=1,\ldots,m$) 
are simple roots, and $:e^{-\alpha_{\bar{i}}\varphi(t_i) }:S_{\bar{i}}(t_i)$ are the screening currents 
defined in \eqref{eq screening currents}. 

Let us calculate the $\varphi$ field correlation and the $\beta\gamma$ correlation 
separately. First, we compute the $\varphi$ field correlation
\begin{equation}
\Psi ({ \bf t, z, \lambda})=\left\langle 
\prod_{i=1}^m :e^{-\alpha_{\bar{i}}\varphi(t_i)  }:  
\prod_{a=1}^n:\exp \left(
\sum_{i=0}^{r_a}  (\lambda^{(a)})_{i}\frac{\partial^{i}\varphi(z_a)}{i!}
 \right)         :
\right\rangle. 
\end{equation}
Recalling the OPE of $\partial^i\varphi(z)$ and $\partial^j\varphi(w)$ 
\eqref{eq free boson d}, we obtain
\begin{align*}
\Psi({\bf t, z,\lambda})=&\prod_{1\le a <b\le n}\left\{
(z_a-z_b)^{\frac{(\lambda^{(a)})_0(\lambda^{(b)})_0}{\kappa}}\right.
\\
&\times\left. 
\exp\left(\sum_{0\le p\le r_a, 0\le q\le r_b, \atop p+q>0}\frac{ (\lambda^{(a)})_p (\lambda^{(b)})_q}
{\kappa}\begin{pmatrix}p+q\\p\end{pmatrix}
\frac{(-1)^{p+1}}{p+q}\frac{1}{(z_a-z_b)^{p+q}}
\right)\right\}
\\
&\times \prod_{1\le i <j\le m}(t_i-t_j)^{\frac{\alpha_{\bar{i}}\alpha_{\bar{j}}}{\kappa}}
 \prod_{i=1}^m\prod_{a=1}^n
\left\{
(t_i-z_a)^{\frac{-\alpha_{\bar{i}}(\lambda^{(a)})_0}{\kappa}}
\exp\left(
\sum_{p>0}^{r_a} \frac{\alpha_{\bar{i}}(\lambda^{(a)})_p}{\kappa}\frac{1}{p(t_i-z_a)^p}
\right)
\right\}. 
\end{align*}

\subsection{}

Next, we compute the $\beta\gamma$ correlation
\begin{equation}
\omega=\left\langle
 \prod_{i=1}^mS_{\bar{i}}(t_i)\prod_{a=1}^nP_a(\gamma(z_a))
\right\rangle
\end{equation}
using the OPEs
\begin{align}
&S_\alpha(z)S_\beta(w)=\frac{1}{z-w}[S_\alpha,S_\beta](w)+\cdots,
\\
&S_\alpha(z)P_a(\gamma(w))=\sum_{p=0}^{r_a}\frac{1}{(z-w)^{p+1}}
\left(S_\alpha[p] P_a\right)(\gamma(w))+\cdots, 
\end{align}
where $\alpha,\beta\in\Delta_+$, $[S_\alpha,S_\beta](w)$
is obtained by replacing $x^\alpha$ with $ \gamma^\alpha(w)$ and 
$\frac{\partial}{\partial x^\alpha}$ with $\beta_\alpha(w) $ in the 
differential operator $[S_\alpha,S_\beta]$, and 
$(S_\alpha[p] P_a)(\gamma(w))$ is obtained by replacing $x^\alpha_i$ with 
$\partial^i \gamma^\alpha(w)/i!$ in 
the polynomial $S_\alpha[p] P_a$. 

Then, we obtain the screening currents Ward identity
\begin{align*}
\omega=&\left\langle
S_{\bar{1}}(t_1)\cdots S_{\bar{m}}(t_m)\prod_{a=1}^nP_a(\gamma(z_a))
\right\rangle
\\
&=\sum_{i=2}^m\frac{1}{t_1-t_i}
\left\langle S_{\bar{2}}(t_2)
\cdots [S_{\bar{1}},S_{\bar{i}}](t_i)\cdots S_{\bar{m}}(t_m)
\prod_{a=1}^nP_a(\gamma(z_a))
\right\rangle
\\
&+\sum_{a=1}^n\sum_{p_a=0}^{r_a}
\frac{1}{(t_1-z_a)^{p_a+1}}
\left\langle
S_{\bar{2}}(t_2)\cdots S_{\bar{m}}(t_m)
P_1(\gamma(z_1))
\cdots (S_{\bar{1}}[p_a]P_a)(z_a)\cdots P_n(\gamma(z_n))
\right\rangle. 
\end{align*}
Thus, we can calculate the $\beta\gamma$ correlation $\omega$ by repeatedly 
using the above relation. 

Therefore, we establish the integral representations of  hypergeometric functions of 
the confluent type 
\begin{equation}
\int \prod_{i=1}^mdt_i\Psi({\bf t,z, \lambda})\omega
\end{equation}
from the WZNW CFT. 

\subsection{}

Let $\g=\slt$. In this subsection, we  write down the $\beta\gamma$ correlation 
$\omega$ and we see that the integral representations coincide with  solutions to the 
confluent KZ equations for $\slt$ given in \cite{JNS}. 

Noting the fact that $(\lambda^{(a)})_p=\lambda^{(a)}_p\alpha/2$ ($a=1,\ldots,n$, $p=0,\ldots,r_a$, 
$\lambda_p^{(a)}\in\C$), 
we see that $\Psi({\bf t,z,\lambda})$ is equal to the master function of the integral 
representation in 
\cite{JNS}. 

Using the screening currents Ward identity, we have 
\begin{equation}
\omega=\langle S(t_1)\cdots S(t_m)\prod_{a=1}^nP_a(\gamma(z_a))\rangle
\\
=\sum_{\sigma\in \frak{S}_m}\sigma\left(\sum_{\left(I_1,\ldots,I_n\right)\in Y}
\prod_{a=1}^nP(a,I_a)
\right), 
\end{equation}
where $\frak{S}_m$ is the set of all permutations of $\{1,\ldots,m\}$ and $Y$, 
a set of 
$(I_1,\ldots,I_n)$ such that $I_j=\{i_1<\ldots<i_k\}$ and $I_1\cup I_2\cup \cdots \cup I_n=\{1,\ldots,m\}$ is a disjoint union, and for $I_a=\{i_1,\ldots,i_k\}$, 
\begin{equation}
P(a,I_a)=\sum_{p_1,\ldots,p_k=0}^{r_a} 
 \frac{\left[ S[p_1]\cdots S[p_k] P_a\right]}{(t_{i_1}-z_a)^{p_1+1}\cdots
 (t_{i_k}-z_a)^{p_k+1}}, 
\end{equation}
where $[f(x)]$ represents the constant term of $f(x)\in\C[x]$. 

For a polynomial $P_a=\langle \lambda^{(a)} | Z f[q_1]\cdots f[q_s]|\lambda^{(a)}\rangle$, 
from the definition of the differential realization \eqref{eq screening current p}, 
we have
\begin{equation}
S[p_1]\cdots S[p_k] P_a=
\langle \lambda^{(a)} | e[p_k]\cdots e[p_1]  Z f[q_1]\cdots f[q_s]|\lambda^{(a)}\rangle. 
\end{equation}
Since the constant term of this polynomial is given by the value at $x_i=0$, that is, $Z=1$, 
we obtain 
\begin{equation}
\left[ S[p_1]\cdots S[p_k] P_a\right]=\langle \lambda^{(a)} | e[p_k]\cdots e[p_1]   f[q_1]\cdots f[q_s]|\lambda^{(a)}\rangle. 
\end{equation}

Therefore,  an element $u$ in $M=M(\lambda^{(1)})\otimes\cdots\otimes M(\lambda^{(n)})$ is
determined by the relation of 
the integral representation and the pairing 
\begin{equation}
\int \prod_{i=1}^mdt_i \Psi({\bf t,z,\lambda}) \omega
=\langle u^*| v_1\otimes\cdots\otimes v_n\rangle, 
\end{equation}
where the element $\langle u^*|$  in the dual module $M^*$  is the dual vector corresponding to 
$u$
and $v_a\in M(\lambda^{(a)})$ ($a=1,\ldots,n$), and this element $u$ 
coincides with the integral representation for the solution of the confluent KZ equation for $\slt$.




\bigskip
\bigskip
\bigskip
\noindent {\it\bf Acknowledgments.}\quad

The authors would like thank to M.~Jimbo  for many helpful discussions.
HN is grateful to T.~Arakawa, K.~Hasegawa, G.~Kuroki, T.~Kuwabara, 
T. Suzuki and Y. Yamada 
for helpful discussions.

\end{document}